\documentclass[aps,pra,twocolumn,superscriptaddress,floatfix,nofootinbib,showpacs,longbibliography]{revtex4-1}
\usepackage[utf8]{inputenc}  
\usepackage[T1]{fontenc}     
\usepackage[british]{babel}  
\usepackage[sc,osf]{mathpazo}
\usepackage[scaled=0.86]{berasans}  
\usepackage[colorlinks=true, citecolor=blue, urlcolor=blue]{hyperref}  
\usepackage{graphicx}
\usepackage{tabularray}
\UseTblrLibrary{amsmath}
\usepackage[babel]{microtype}  
\usepackage{amsmath,amssymb,amsthm,bm,amsfonts,mathrsfs,bbm} 

\usepackage{xspace}  
\usepackage{pgf,tikz}
\usepackage{xcolor}
\usepackage{multirow}
\usepackage{array}
\usepackage{bigstrut}
\DeclareUnicodeCharacter{020B}{\^{\i}}
\usepackage{braket}
\usepackage{color}
\usepackage{natbib}
\usepackage{multirow}
\usepackage{mathtools}
\usepackage{float}
\usepackage{xcolor,colortbl}
\usepackage{physics}
\usepackage{amsmath}
\usepackage{color}
\usepackage{bbold}
\usepackage[justification=justified, format=plain]{subcaption}
\usepackage[justification=raggedright]{caption}

\newcommand{\be}{\begin{equation}}
\newcommand{\ee}{\end{equation}}
\newcommand{\ba}{\begin{eqnarray}}
\newcommand{\ea}{\end{eqnarray}}

\newtheorem{theorem}{Theorem}

\newtheorem{definition}{Definition}

\newtheorem{example}{Example}

\def\>{\rangle}
\def\<{\langle}

\usepackage{comment}

\begin{document}

\title{Optical interferometric tools for probing memory effects}

\author{Saheli Mukherjee}
\email{mukherjeesaheli95@gmail.com}
\affiliation{S. N. Bose National Centre for Basic Sciences, Block JD, Sector III, Salt Lake, Kolkata 700 106, India}

\author{Bivas Mallick}
\email{bivasqic@gmail.com}
\affiliation{Statistics and Mathematics Unit, Indian Statistical Institute, 560059 Bangalore,
India}

\author{Samyadeb Bhattacharya}
\affiliation{Center for Quantum Science and Technology,
International Institute of Information Technology, Gachibowli, Hyderabad 500 032, India} 

\author{Ananda G. Maity}
\email{anandamaity289@gmail.com}
\affiliation{School of Physical Sciences, Indian Institute of Technology Goa, Ponda 403401, Goa, India}

\author{A. S. Majumdar}
\affiliation{S. N. Bose National Centre for Basic Sciences, Block JD, Sector III, Salt Lake, Kolkata 700 106, India}

\begin{abstract}
Recent efforts to detect non-Markovian memory have become increasingly important, motivated by the broad relevance of memory effects in quantum information processing, metrology, and thermodynamic engines. Here, we show that interferometric methods, particularly the phase shift and visibility, provide an experimentally accessible framework for probing environmental effects and witnessing non-Markovian dynamics. Specifically, suitably constructed witness operators enable the detection of non-Markovianity through measurable phase shifts in the interference pattern. Furthermore,  we prove that the interferometric visibility is monotonically non-increasing under incoherent CP-divisible (Markovian) dynamics. Consequently, any revival of the optimized visibility constitutes a sufficient signature of information backflow and hence non-Markovianity. The effectiveness of the proposed phase-based and visibility approaches is illustrated for representative examples of non-Markovian dynamics. Our results offer a simple and experimentally feasible framework for probing memory effects, with potential applications in quantum sensing and noisy intermediate-scale quantum devices.

\end{abstract}
\maketitle

\section{Introduction}\label{s1}
Understanding memory effects in open quantum systems is a central challenge in quantum physics, especially for harnessing noise as a resource in quantum technologies \cite{xiang2014entanglement,reich2015exploiting}, quantum information processing \cite{laine2014nonlocal,PhysRevLett.100.180402,bylicka2014non,zhang2024non}, metrology \cite{PhysRevLett.127.060501}, and thermodynamics \cite{PhysRevE.97.062108,marcantoni2017entropy,bylicka2016thermodynamic,bhattacharya2020thermodynamic,PhysRevA.111.032416}. In realistic settings, quantum systems are inevitably coupled to their environments, leading to decoherence, dissipation and other irreversible phenomena. When this interaction is weak and the environment has no memory, the system dynamics can be described by memoryless (Markovian) processes, characterized by a monotonic loss of quantum features such as coherence and distinguishability. However, many physically relevant environments exhibit memory effects, giving rise to non-Markovian dynamics, where information can flow back from the environment to the system \cite{rivas2014quantum,RevModPhys.88.021002}. 

In recent times, such non-Markovian memory has emerged as a valuable resource for quantum technologies, enabling advantages in tasks such as quantum teleportation with mixed states \cite{laine2014nonlocal}, entanglement distribution \cite{xiang2014entanglement}, communication through long quantum channels \cite{bylicka2014non}, and enhanced work extraction in quantum thermal machines \cite{PhysRevE.97.062108}. These advantages stem from information backflow from the environment to the system, whereby non-Markovian memory is transformed into useful resources such as entanglement, coherent information, and extractable work. This highlights the importance of identifying and characterizing memory effects in open quantum dynamics, both for understanding their physical origin and for determining when they can be exploited as a resource.

Detecting and quantifying non-Markovianity remains a nontrivial task, and a variety of witnesses have been proposed based on information backflow \cite{PhysRevLett.103.210401}, divisibility of dynamical maps \cite{rivas2010entanglement,PhysRevA.109.022247}, revivals of quantum correlations \cite{PhysRevA.85.032318,zhang2026analytical,chanda2016delineating}, uncertainty relations \cite{maity2020detecting}, mutual information between system and ancilla \cite{PhysRevA.86.044101}, accessible information \cite{PhysRevLett.112.210402}, etc. Among these, coherence plays a particularly important role. Since coherence is directly linked to interference phenomena, it provides a natural operational route to probe memory effects \cite{chanda2016delineating, zhi2015non,PhysRevA.96.022106,shao2020quantifying, sun2022quantifying}. This connection is especially appealing in interferometric settings, where coherence is manifested directly through the visibility of interference fringes.

Interferometric visibility quantifies the contrast of interference fringes and hence can serve as an operational probe of the nature of noisy quantum dynamics through its direct relationship with quantum coherence \cite {biswas2017interferometric}. Under Markovian noise, coherence, and hence visibility typically decays monotonically. In contrast, the presence of environmental memory can induce revivals of coherence, which manifest as non-monotonic behavior in the visibility. This establishes a direct connection between experimental quantities like interferometric visibility, phase and the underlying dynamical properties of the system.
In this work, we explore quantum interferometric visibility as a physically transparent and experimentally accessible witness of non-Markovian dynamics.

The connection between non-Markovianity and interferometry has previously been explored primarily through a metrological measure of quantum correlations dubbed as quantum interferometric power, which is defined as the minimum quantum Fisher information over all local Hamiltonians acting on an ancillary subsystem \cite{PhysRevA.91.032115}. This approach was shown to consistently capture information backflow for dephasing and amplitude-damping channels, and was  extended to continuous-variable systems \cite{PhysRevA.92.052122}. The significance of quantum interferometric power as a useful quantifier of quantum correlations in a two-qubit system subjected to a common dephasing bath has been studied in \cite{PhysRevA.105.042612}. A quantifier of non-Markovian dynamics with and without the rotating-wave approximations in system-environment interactions has been modelled \cite{xiong2023quantum}. Nonclassicality under non-Markovian random telegraph noise has been quantified using the interferometric power \cite{elghaayda2023quantum}, but the analysis is restricted to a physical model of a bipartite Gisin state \cite{gisin1996hidden, PhysRevLett.74.2619,thirring2011entanglement}.
While the above approach establishes a fundamental connection between information backflow and metrological correlations, its implementation requires an ancillary subsystem and  a nonlinear optimization over local Hamiltonians.  Unlike these approaches based on the quantum interferometric power, our objective is to identify a witness that relies solely on the interference pattern of a single evolving quantum system, thereby avoiding ancillary resources and optimization procedures.

In this work we show that the dynamical behavior of interferometric visibility provides a robust and operational criterion to detect non-Markovianity. By linking visibility to non-Markovianity measures, we further demonstrate how our approach captures the flow of quantum information between system and environment. Our framework thus offers a conceptually simple and experimentally feasible method to witness non-Markovian dynamics, with potential applications in quantum sensing and noisy intermediate-scale quantum devices.

The paper is organized as follows. In Sec.~\ref{secII}, we provide a brief overview of the dynamics of open quantum systems, followed by a discussion of quantum coherence and its quantification using the $l_1$-norm. We also present the basic Mach–Zehnder interferometer setup along with its constituent elements. In Sec.~\ref{secIII}, we introduce our main formalism for detecting non-Markovian dynamics through an interference pattern. In particular, we develop witness-based and coherence-based detection schemes utilizing the phase shift and visibility of the interference pattern. We consider representative examples of pure dephasing and depolarizing dynamics to illustrate the applicability of our approach. Finally, we summarize our main findings along with some future perspectives in Sec.~\ref{secIV}.

\section{Preliminaries} \label{secII}
In this section, we outline the basic concepts, notation, and mathematical tools underlying the results presented in this work. Throughout the paper, we adopt the standard formalism and terminology of quantum information theory.

\subsection{Open quantum system dynamics} \label{secIIA}
Interactions with the environment are ubiquitous in nature. While isolated quantum systems evolve unitarily, for systems interacting with the environment (commonly known as open quantum systems), the dynamics is governed by completely positive and trace-preserving (CPTP) maps. Linear CPTP maps are linear transformations between density operators, i.e., $\Lambda(t,t_0):\rho(t_0)\rightarrow \rho(t)$. We assume that the inverse map $\Lambda(t,t_0)^{-1}$ exists for all times $t \ge t_0$. A dynamical map $\Lambda(t,t_0)$ is said to be \textit{divisible} iff, for any intermediate time $s$ satisfying $t\geq s\geq t_0$, it can be decomposed as
\begin{equation}
    \Lambda(t,t_0)=\Lambda(t,s) \circ \Lambda(s,t_0) \label{divisibility}
\end{equation}
where $'\circ'$ denotes the map composition. Although the initial map $\Lambda(s,t_0)$ and the overall map $\Lambda(t,t_0)$ are completely positive, the intermediate map $\Lambda(t,s)$ need not represent a physical dynamical evolution. The dynamics $\Lambda(t,t_0)$ is said to be \textit{positive divisible} (P-divisible) if the map $\Lambda(t,s)$ is positive and the divisibility property, as given by Eq.~\eqref{divisibility} is satisfied for every $t\ge s \ge t_0$. The dynamics $\Lambda(t,t_0)$ is said to be \textit{completely positive divisible} (CP-divisible) if the map $\Lambda(t,s)$ is completely positive (CP) and Eq.~\eqref{divisibility} is satisfied for every $t\ge s \ge t_0$. Note that by virtue of the Choi-Jamiolkowski isomorphism \cite{choi1975completely, jamiolkowski1972linear}, complete positivity of the map $\Lambda(t,s)$ implies that the associated Choi operator $\mathcal{C}_{\Lambda(t,s)} = [\mathbb{I} \otimes \Lambda(t,s)] \ket{\Phi^+}\bra{\Phi^+}$ (where $\ket{\Phi^+}$ denotes the maximally entangled state) is positive. Depending on the nature of the system-environment interaction, the dynamics $\Lambda(t,t_0)$ is termed Markovian if it is divisible over all time intervals; otherwise, it is non-Markovian. This characterization of non-Markovian dynamics is commonly referred to as the RHP notion of non-Markovianity \cite{rivas2010entanglement}.

Another widely used approach to characterize Markovianity, introduced by Breuer, Laine, and Piilo (BLP) is based on the distinguishability between quantum states which is generally quantified by the trace distance \cite{PhysRevLett.103.210401}. Under Markovian evolution, the trace distance decreases monotonically due to the continuous loss of information from the system to the environment. Any increase in the trace distance indicates a backflow of information from the environment to the system and is therefore identified as a signature of non-Markovian behavior. This provides an operationally meaningful and physically intuitive description of memory effects in open quantum systems.

\subsection{Quantum coherence and its measure}
Quantum coherence arises due to the superposition principle and is characterized by the presence of off-diagonal terms in the density matrix for a fixed chosen basis. Any state of the form 
\begin{equation}  \label{incoherent}
    \rho=\sum_{i=0}^{d-1}c_i \ket{i}\bra{i} 
\end{equation} 
is diagonal with respect to the basis $\{\ket{i}\}_{i=0}^{d-1}$ and thus represents an incoherent state. States that cannot be written in the above form are said to be coherent. Such states are useful resources with applications in quantum thermodynamics \cite{PhysRevLett.113.150402,narasimhachar2015low,PhysRevX.5.021001,korzekwa2016extraction,scully2011quantum}, metrology \cite{giovannetti2004quantum,PhysRevA.94.052324}, quantum algorithms \cite{PhysRevA.93.012111,anand2016coherence}, channel discrimination \cite{PhysRevLett.116.150502,PhysRevA.93.042107}, motivating the development of the resource theory of coherence \cite{PhysRevLett.113.140401,levi2014quantitative,PhysRevLett.116.120404,RevModPhys.89.041003}. 
In this resource theory, incoherent states (as defined in Eq.~\eqref{incoherent}) constitute the free states, and incoherent operations are the free operations. 

To quantify coherence present in a state (if any), several measures have been proposed \cite{PhysRevA.93.042107,aberg2006quantifying,PhysRevLett.113.140401,PhysRevA.92.022124,PhysRevLett.115.020403,PhysRevA.91.042120,zhang2017estimation,PhysRevA.92.052114,biswas2017interferometric,PhysRevLett.116.070402}. Here, we focus on the $l_1$-norm of coherence \cite{PhysRevLett.113.140401}, defined as
\begin{equation}  \label{coherencemeasure}
    C_{l_1}(\rho)=\sum_{i \neq j}|\rho_{ij}|
\end{equation}

For a qubit density matrix $\rho=\begin{pmatrix}
\rho_{00} & \rho_{01} \\
\rho_{10} & \rho_{11}
\end{pmatrix}$, the expression for $l_1$ norm reduces to $C_{l_1}(\rho)=2|\rho_{01}|$.
It has been further shown that the $l_1$ norm of coherence satisfies positivity, convexity, and monotonicity under incoherent operations, thereby constituting a valid coherence measure. 
For a maximally coherent state of dimension $d$, the value of this measure equals $d-1$.

In addition, quantum coherence is also intrinsically connected to several other non-classical features such as quantum interference \cite{young1802ii}, with coherence being a necessary prerequisite for the observation of interference effects. Motivated by this connection, several measures of coherence in terms of interferometric visibility have also been proposed \cite{PhysRevA.92.012118,PhysRevLett.116.160406,biswas2017interferometric,PhysRevLett.118.020403}. To understand this further, we briefly review the basic setup of the Mach-Zehnder interferometer below.

\subsection{Details of the interferometric setup}
This subsection is devoted to a detailed understanding of a Mach-Zehnder interferometer setup. Figure \ref{Pic1} represents a schematic diagram of the setup.
\begin{figure}[ht]
\includegraphics[width=0.5\textwidth]{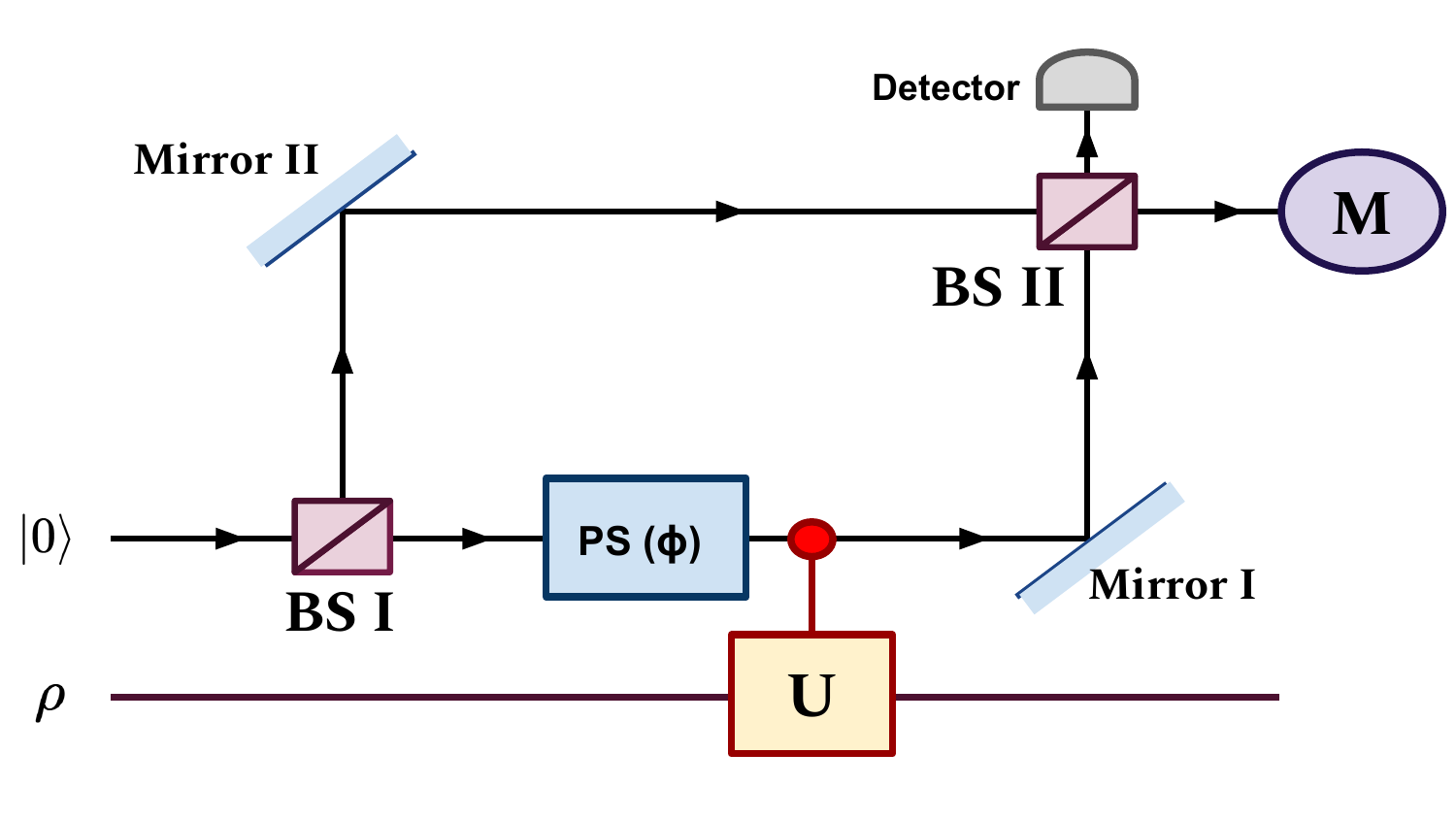}
\caption{The interferometer setup consists of two $50-50$ beam splitters (denoted by BS I and BS II), a phase shifter ($PS(\phi)$), two mirrors (Mirrors I and II) and a detector. $\ket{0}$ and $\rho$ represent the inputs to the path and internal dof. $U$ represents a controlled unitary acting on the internal dof as target. $M$ is a measurement performed along path $\ket{0}$ in the computational basis.}\label{Pic1}
\centering
\end{figure}
Let $\mathcal{H}$ and $\mathcal{H}'$ represent the Hilbert spaces corresponding to the path and internal degrees of freedom (dof) of the interferometer, respectively. The path Hilbert space is spanned by the basis ${\ket{0},\ket{1}}$, representing the two interferometric paths, while ${\ket{k}}_{k=1}^N$ forms a basis for the internal dof. We consider an initial input state of the form $\ket{0}\bra{0}\otimes\rho$, where $\rho$ acts on the internal Hilbert space and is taken to be diagonal in the basis ${\ket{k}}$. 
This state then undergoes a sequence of unitary operations, as described below. 
\begin{itemize}
    \item A Hadamard gate $H$ acting on the path degree of freedom, which represents the first $50-50$ beam splitter (BS I). The corresponding unitary is $H\otimes\mathbb{I}$.
    \item The Hadamard is followed by a phase shifter (PS ($\phi$)) introducing a relative phase $\phi$ in the path $\ket{0}$, described by the unitary
    $e^{i \phi} \ket{0}\bra{0} \otimes \mathbb{I} + \ket{1}\bra{1} \otimes \mathbb{I}$.
    \item A controlled unitary operation with the path degree of freedom as the control and the internal degree of freedom as the target, is given by $\ket{0}\bra{0} \otimes \mathbb{I} + \ket{1}\bra{1} \otimes U$.
    \item Mirrors are represented by $\sigma_{x} \otimes \mathbb{I}$.
    \item A second Hadamard operation, representing the second beam splitter BS II is subsequently applied to the path degree of freedom with corresponding unitary: $H \otimes \mathbb{I}$.
    \item Finally, a projective measurement is performed in the $\{\ket{0},\ket{1}\}$ basis.
\end{itemize}

The overall evolution of the interferometer yields the final state, $\rho_{f} = U_{total} (\rho \otimes \ket{0}\bra{0}) U_{total}^{\dagger}$, which after simplication reduces to:
\begin{equation}
\begin{split}
\rho_{f} =& \frac{1}{2} (\ket{+}\bra{+} \otimes U \rho U^{\dagger} + e^{-i \phi} \ket{+}\bra{-}\otimes U \rho \\ &+ e^{i \phi} \ket{-}\bra{+}\otimes \rho U^{\dagger} + \ket{-}\bra{-}\otimes \rho). \label{finalstate}
\end{split}
\end{equation}
The intensity measured at the output port corresponding to the path $\ket{0}$ is given by $\mathcal{I}=\Tr[(\ket{0}\bra{0}\otimes \mathbb{I}) \rho_{f}]$. Using $\Tr(U \rho) = |\Tr(U \rho)| e^{i \text{arg}(Tr(U \rho))}$ and $\Tr[(U \rho)] = \Tr[\rho U^{\dagger}]$, one can obtain
\begin{equation}
    \mathcal{I} = \frac{1}{2} [1+ |\Tr(U \rho)| \cos{(\phi - \text{arg}[\Tr(U \rho)])}]  \label{intensity}
\end{equation}
where $|\Tr(U \rho)|$ and $\alpha=\text{arg}[\Tr(U \rho)]$ are the visibility and phase shift of the interference pattern respectively. These represent the primary quantities of interest of an interferometer and are useful for detecting many resources \cite{PhysRevLett.85.2845,PhysRevLett.89.127902,kanjilal2023entanglement,PhysRevLett.88.217901,PhysRevA.92.022316,PhysRevLett.125.123601, 828t-3czf}. 

Below, we show how to detect a non-Markovian dynamics from the visibility and phase shift of an interference pattern by choosing a suitable unitary $U$.  

\section{Detection of non-Markovian dynamics in an interferometric setup} \label{secIII}
Before proceeding to the detection scheme, let us review the working principle of a Mach–Zehnder interferometer with two paths and examine the crucial role played by noise introduced in the input state. In the absence of any noise, we can design the interferometric setup by adjusting the path lengths and phase such that one detector experiences complete constructive interference while the other experiences complete destructive interference. However, if noise is introduced into the system, the quality of interference will degrade. This may cause photons to sometimes appear at the detector meant for destructive interference. This phenomenon can be explained by the well-known complementarity between the "which-path information of the photon" and the "interferometric visibility": when noise is present, the photon leaves some trace on the interferometer's arm, making it possible, in principle, to determine which path the photon has gone through. This leads to a reduction in the visibility of the interference pattern. Thus, in the presence of noise, the quantum nature of the interference would effectively reduce towards a classical probabilistic path. So, whenever a photon is detected at the previously destructive port, it indicates the presence of noise, providing a footprint of the disturbance. 

However, apart from the mere presence of noise, its nature also has a direct influence on the interference pattern. Non-Markovian noise introduced into the system will induce a backflow of information, and hence the visibility will exhibit time-dependent oscillations, from which we can potentially detect the signature of non-Markovianity.
Building on this concept, we aim to detect the non-Markovianity present in the input state by proposing witness-based and coherence-based detection schemes from the phase shift and visibility (respectively) of an interference pattern in the subsequent sections. 

\subsection{Witness-based detection of non-Markovian dynamics from the phase shift of an interference pattern}\label{s2}
 A convex resource theory of non-Markovianity under the constraint of small time approximation
has been proposed earlier \cite{bhattacharya2021convex}. Note that such small time approximation is introduced to restore the convexity of the set of CP-divisible free operations, which is otherwise non-convex \cite{PhysRevLett.101.150402,wolf2008dividing}. Using the Choi-Jamiolkowski isomorphism \cite{jamiolkowski1972linear,choi1975completely}, Markovian dynamics are mapped to Choi states that constitute the free states of the theory, while Choi states corresponding to non-Markovian dynamics represent resourceful states. In the limit $\Delta t\to0$,  the set of Choi states associated with Markovian dynamics, ${\mathcal{C}_{\mathcal{M}}(t+\Delta t,t)}$, forms a convex and compact set. Consequently, by the Hahn-Banach separation theorem \cite{holmes2012geometric}, there exists a hyperplane capable of separating at least one non-Markovian Choi state from the set of all Markovian ones, i.e., if $W$ represents the witness operator, 
\begin{equation*}
\begin{split}
    & \Tr[W \mathcal{C} (t+\Delta{t},t)] \ge 0 ~~ \forall \mathcal{C}(t+\Delta{t},t) \in \{\mathcal{C}_{\mathcal{M}} (t+\Delta{t},t)\},
    \\ & \text{and,} \Tr[W \mathcal{C}^{*} (t+\Delta{t},t)] < 0 ~~\text{for at least one} ~~ \\& \mathcal{C}^{*} \notin \{\mathcal{C}_{\mathcal{M}} (t+\Delta{t},t)\}.
    \end{split}
\end{equation*}
Below, we present one of our main theorem to detect a non-Markovian dynamics from the phase-shift of an interference pattern, using the framework of witness operators.

\begin{theorem} \label{theorem1}
    Every linear non-Markovian witness operator admits an interferometric realization through the phase shift of an interference pattern.
\end{theorem}
\begin{proof}
    Corresponding to a dynamical map $\Lambda$, henceforth we will denote its Choi state by $\mathcal{C}_{\Lambda}$. We take the input state of the interferometer to be $\rho=\mathcal{C}_{\Lambda}$ and choose the controlled unitary as $U=e^{i \theta W}$, where $W$ is a non-Markovian witness operator, and $\theta\in\mathbb{R}^+$. For an infinitesimal value of $\theta$,
\begin{equation} \label{visibility}
\begin{split}
\Tr(U \rho) = & \Tr(e^{i \theta W}\mathcal{C}_{\Lambda}) \simeq \Tr(\mathcal{C}_{\Lambda}+i \theta W \mathcal{C}_{\Lambda}) \\ & =1+i\theta\Tr(\mathcal{C}_{\Lambda} W) \simeq e^{i \theta \Tr(W \mathcal{C}_{\Lambda})}.
\end{split}
\end{equation}

In this case, the phase shift is given by $\text{arg}[\Tr(U \mathcal{C}_{\Lambda})] = \theta \Tr(W \mathcal{C}_{\Lambda})$. For Markovian dynamics, the witness satisfies $\Tr(W\mathcal{C}_{\Lambda})\geq0$, leading to an effective phase shift in Eq.~\eqref{intensity} from $\phi$ to $\phi-\theta\Tr(W\mathcal{C}_{\Lambda})$. In contrast, if $\Tr(W\mathcal{C}_{\Lambda})<0$, this results in a phase shift of the form $\phi+\theta|\Tr(W\mathcal{C}_{\Lambda})|$, which is a signature of non-Markovianity.  
\end{proof}Hence, the nature of the underlying dynamics can be inferred directly from the observed phase shift in the interference pattern. It is worth noting that, 
the interferometric visibility defined in \eqref{visibility} can be accessed through suitable correlation measurements, without requiring the direct preparation or complete reconstruction of the intermediate Choi operator, as demonstrated in Ref. \cite{roy2026witness}.

Below we provide explicit examples to support our above theorem.

\begin{example} \label{example1}
    Pure dephasing dynamics: Consider the pure dephasing master equation~\cite{rivas2014quantum}
\begin{equation}
    \frac{d \rho (t)}{dt} = \mathcal{L}(\rho(t)) = \gamma(t) (\sigma_{z} \rho(t) \sigma_{z} - \rho(t) ) \label{dephasing}
\end{equation}
where $\gamma$ is the Lindblad coefficient. The dynamics is Markovian for $\gamma(t)\geq0$ and non-Markovian whenever $\gamma(t)<0$. Under the small time approximation (i.e., for $|\Delta t \gamma| <<1$), the corresponding dynamical map can be written as $\Lambda (\rho)= e^{\mathcal{L}} (\rho) = (\mathbb{I}+ \Delta t \mathcal{L} ) \rho$. The associated Choi state is given by 
\begin{equation}
    \mathcal{C}_{\Lambda}=(\mathbb{I} \otimes (\mathbb{I}+ \Delta t \mathcal{L} )) \ket{\Phi^{+}}\bra{\Phi^{+}} \label{Choi}
\end{equation}
where $\ket{\Phi^{+}}= \frac{\ket{00}+\ket{11}}{\sqrt{2}}$ is the maximally entangled state. The eigenvalues of $\mathcal{C}_{\Lambda}$ are $0, 0, \gamma \Delta t$ and $(1-\gamma \Delta t)$. For $\gamma < 0, \gamma \Delta t<0$, and hence the Choi operator loses its positive semi-definiteness, indicating the signature of non-Markovianity.  The witness of this non-Markovian dynamics can be constructed by taking the projector of the eigen vector corresponding to this negative eigen value of the Choi operator, i.e., $W = \begin{pmatrix}
1 & 0 & 0 & -1 
\end{pmatrix}^{T}$ $\begin{pmatrix}
1 & 0 & 0 & -1 
\end{pmatrix}$ \cite{bhattacharya}. 
For this witness, 
$\Tr(W \mathcal{C}_{\Lambda}) = \gamma \Delta t$. The phase shift is given by 
\begin{equation*}
    \phi-\theta \Tr(W \mathcal{C}_{\Lambda}) = \phi-\theta \gamma \Delta t
\end{equation*}
For all Markovian dynamics, $\gamma \Delta t \ge 0$, and the phase shift reduces from $\phi$ to $\phi-\theta \gamma \Delta t$, whereas an increase in the phase shift from $\phi$ to $\phi+\theta \gamma \Delta t$ is a signature of non-Markovianity.
\end{example} 
\begin{example} \label{example2}
    Depolarizing dynamics: For a depolarizing dynamics, the corresponding master equation is given by \cite{rivas2014quantum}
\begin{equation}
\begin{split}
    \frac{d \rho(t)}{dt} = &\mathcal{L}(\rho(t)) = \gamma_{x}(t) (\sigma_{x} \rho(t) \sigma_{x} - \rho(t) ) \\ & + \gamma_{y}(t) (\sigma_{y} \rho(t) \sigma_{y} - \rho(t) ) \\ & + \gamma_{z}(t) (\sigma_{z} \rho(t) \sigma_{z} - \rho(t) ).\label{depolarizing}
    \end{split}
\end{equation} 
 where $\gamma_{i}(t)$ with $i \in \{x,y,z\}$ are the Lindblad coefficients. The Choi operator has eigenvalues $1- \Delta t (\gamma_{x}(t)+\gamma_{y}(t)+\gamma_{z}(t)), \gamma_{x}(t)\Delta t, \gamma_{y}(t)\Delta t,\gamma_{z}(t)\Delta t$. The dynamics is non-Markovian if $\gamma_{i}(t) < 0$ for at least one $i \in \{x,y,z\}$. Witness can be constructed from the projector corresponding to any one of the negative eigenvalues $\gamma_{i}(t)$. For our analysis, we consider the case where $\gamma_x(t)<0$ and construct the witness operator onto the eigenvector corresponding to the negative eigenvalue $\gamma_x(t)\Delta t$, namely
$W = \begin{pmatrix}
0 & 1 & 1 & 0 
\end{pmatrix}^{T}$ $\begin{pmatrix}
0 & 1 & 1 & 0 
\end{pmatrix}$ \cite{bhattacharya}. 
For this witness, 
$\Tr(W \mathcal{C}_{\Lambda}) = \gamma_{x} \Delta t$. The phase shift is given by 
\begin{equation*}
    \phi-\theta \Tr(W \mathcal{C}_{\Lambda}) = \phi-\theta \gamma_{x} \Delta t
\end{equation*}
As discussed earlier, there is a reduction in phase shift for all Markovian dynamics, whereas an increase in the phase shift indicates non-Markovianity.

It is worthwhile to mention that for $\gamma_{x}(t)=\gamma_{y}(t)=1$ and $\gamma_{z}(t)=- \tanh{t}$, this dynamics shows eternal non-Markovianity \cite{hall2014canonical}. Hence, this dynamics cannot be detected by various measures of non-Markovianity \cite{bhattacharya}. But the witness mentioned above is able to detect this dynamics. Therefore, our detection approach based on phase shift is able to detect eternal non-Markovianity, thus highlighting the significance of using interferometry for detecting non-Markovianity.
\end{example}

Apart from the phase shift, the visibility of an interference pattern is another quantity of interest, as discussed earlier. We now explore the significant role played by visibility in the detection of non-Markovian dynamics.
\subsection{Interferometric visibility as a signature of non-Markovian dynamics}
In this section, we propose a criterion for the detection of non-Markovian dynamics based on interferometric visibility. The criterion is satisfied by all incoherent Markovian dynamics, whereas a violation indicates non-Markovianity. Before proceeding further, we formally introduce incoherent Markovian dynamics. 

\textit{Incoherent Markovian dynamics} \cite {chanda2016delineating}: Consider the dynamical evolution of a system governed by the master equation \cite{gorini1976completely,lindblad1976generators}
\begin{equation}  \label{masterequation}
\begin{split}
    \dot{\rho}(t)=&-i[H(t), \rho(t)]+\sum_{k}\gamma_{k}(t)[A_{k}(t)\rho(t)A^{\dagger}_{k}(t)\\ &-\frac{1}{2}\{A^{\dagger}_{k}(t)A_{k}(t), \rho(t)\}], 
    \end{split}
\end{equation}
where $H(t)$, $A_{k}(t)$ and $\gamma_{k}(t)$ denote the Hamiltonian of the system, the Lindblad operators and the Lindblad coefficients, respectively. The condition for incoherence with respect to a basis demands that 
\begin{equation*}
    \dot{\rho}_{ij}(t)=0 \hspace{2mm} \forall i \neq j
\end{equation*}
This implies that in Eq.~\eqref{masterequation}, $H(t)$ is diagonal in that basis and $A_{k}(t)$ can have at most one non-zero element in each column in that basis \cite{chanda2016delineating}. The condition for Markovianity implies $\gamma_{k}(t) \ge 0$ at all times $t$. Therefore, if the Lindblad coefficient becomes negative at some particular instant, then that is a sign of non-Markovian dynamics. 

Next, we demonstrate that non-Markovianity can be detected from the visibility by choosing $U$ from a suitable class of interferometric unitaries, denoted by $\{\mathcal{U}_
{int}\}$. The corresponding visibility over this admissible class is defined as follows.
\begin{definition}
    The visibility of an interference pattern under the interferometric class of unitaries is given by 
    \begin{equation}
    \begin{split}
\tilde{V}(t) & =  \max\limits_{U \in \mathcal{U}_{int}} \left|\operatorname{Tr}(U \rho)\right| \\ & = |\Tr(U_{opt} \rho)|, 
\end{split}
    \end{equation}
    where 
    \begin{equation} \label{Uint}
        \mathcal{U}_{\mathrm{int}}
=
\left\{
\begin{pmatrix}
0 & e^{i\alpha} \\
e^{i\beta} & 0
\end{pmatrix}
\,\middle|\,
\alpha,\beta \in \mathbb{R}
\right\},
 \end{equation} and $U_{opt}$ denotes the optimal unitary.
$U_{int}$ represents the class of interferometrically implementable path-exchange unitaries. In a two-path Mach-Zehnder interferometer, this controlled operation acts by exchanging the two paths while allowing arbitrary phases. This operational restriction eliminates any contribution from the path populations, thus ensuring that the resulting visibility depends on the coherence of the quantum states.  
\end{definition}
We now formulate the detection of non-Markovianity from the visibility in the following theorem.
\begin{theorem} \label{theorem2}
   For any incoherent Markovian (CP-divisible) noise acting between the two beam splitters of a Mach-Zehnder interferometer, the interferometric visibility maximized over the set of all interferometric unitaries $\{\mathcal{U}_{int}\}$ is a monotonically non-increasing function of time:
   \begin{equation}
       \frac{d}{dt} \tilde{V}(t) \leq 0 ~~\forall ~t.
   \end{equation}
   Any violation of this monotonicity (i.e. $\dot{\tilde{V}}(t)>0$ at some time $t$) is sufficient to conclude that the dynamics is non-Markovian.
\end{theorem}
\begin{proof}
    We use the BLP measure of non-Markovianity for our purpose. We begin by establishing a relation between the BLP measure and quantum coherence for qubit dynamics. 

         Consider two arbitrary states $\rho_1(t)=
\begin{pmatrix}
\rho_{00}^{(1)}(t) & \rho_{01}^{(1)}(t) \\[6pt]
\rho_{10}^{(1)}(t) & \rho_{11}^{(1)}(t)
\end{pmatrix}$
and
$\rho_2(t)=
\begin{pmatrix}
\rho_{00}^{(2)}(t) & \rho_{01}^{(2)}(t) \\[6pt]
\rho_{10}^{(2)}(t) & \rho_{11}^{(2)}(t)
\end{pmatrix}$. 

\bigskip
The distance between them is given by \begin{equation*}
\begin{split}
    D(\rho_{1}(t), \rho_{2}(t))& = \frac{1}{2}\Vert \rho_{1}(t)-\rho_{2}(t) \Vert \\ & = \frac{1}{2}\Tr[\sqrt{(\rho_1(t)-\rho_2(t))^{\dagger}(\rho_1(t)-\rho_2(t))}]\\ & = \frac{1}{2}\Tr[\sqrt{\Delta(t)^{\dagger} \Delta(t)}], 
    \end{split}
\end{equation*}
where $\Delta(t) = \rho_{1}(t)-\rho_{2}(t)$ is a Hermitian and traceless matrix. It can be represented in a compact form as 
\begin{equation*}
    \Delta(t)= \begin{pmatrix}
        \delta(t) & \beta(t) \\
        \beta(t)^{*} & -\delta(t) \\
    \end{pmatrix},
\end{equation*} where $\delta(t)=\rho_{00}^{(1)}(t)-\rho_{00}^{(2)}(t)$ and $\beta(t)=\rho_{01}^{(1)}(t)-\rho_{01}^{(2)}(t).$\\
The distance then reduces to
$D(\rho_{1}(t), \rho_{2}(t)) =\sqrt{|\beta (t)|^2 + |\delta (t)|^2}$. $|\delta(t)|=\frac{1}{2} \sum_{i} |\langle i| \Delta(t)| i \rangle|$ and $C_{l_1}(\Delta(t))=\sum_{i \neq j}|\Delta_{ij}(t)|=2 |\beta(t)|.$ Thus, $|\delta(t)|$ represents the classical trace distance between the populations (and hence is denoted by $D_{cl}$ henceforth) and $|\beta(t)|$ is related to the coherence. Thus, 
\begin{equation*}
\begin{split}
   & [D(\rho_{1}(t), \rho_{2}(t))]^2 = D^2_{cl}+ \frac{1}{4} [C_{l_1}(\Delta(t))]^2. \\ & \frac{d}{dt}D(\rho_1(t),\rho_{2}(t))=\frac{1}{D} [D_{cl} \dot{D}_{cl}+\frac{1}{4}C_{l_1}(\Delta(t)) \dot{C}_{l_1}(\Delta(t))], 
    \end{split}
\end{equation*}
where we have used the shorthand notation $D$ for $D(\rho_{1}(t), \rho_{2}(t))$ and $\dot{x}=\frac{dx}{dt}$ for a time-dependent function $x$.
The BLP measure is now given by 
\begin{equation} \label{NMandcoherence}
    \begin{split}
        \mathcal{B} &= \max_{\rho_1(0), \rho_{2}(0)} \int_{\frac{d}{dt} D>0} \frac{d}{dt}D(\rho_1(t),\rho_{2}(t)) dt \\ & = \max_{\rho_1(0),\rho_{2}(0)} \int_{\dot{D}>0} \frac{1}{D}[D_{cl} \dot{D}_{cl}+\frac{1}{4}C_{l_1}(\Delta(t)) \dot{C}_{l_1}(\Delta(t))]dt.
    \end{split}
\end{equation}
where $\rho_{{1}(2)}(0)$ are the initial states. 

For a special pair of states where $\rho_{2}(t)=\sigma_{z}\rho_{1}(t)\sigma_z$, $D_{cl}$ vanishes. These kind of transformations preserve the populations and change only the phases of the off-diagonal elements.  For this special case, $D=\frac{1}{2} C_{l_1}(\Delta(t))= C_{l_1}(\rho_1(t))= C_{l_1}(\rho_2(t))$. Eq.~\eqref{NMandcoherence} then reduces to 
\begin{equation} \label{equality}
\begin{split}  
    \mathcal{B}=& \max_{\rho_1(0)} \int_{\frac{d}{dt} C_{l_1}(\rho_1(t))>0} \frac{d}{dt} C_{l_1}(\rho_1(t)) dt \\ & = \max_{\rho_2(0)} \int_{\frac{d}{dt} C_{l_1}(\rho_2(t))>0} \frac{d}{dt} C_{l_1}(\rho_2(t)) dt.
    \end{split}
\end{equation}
This establishes a direct relation between the BLP measure and the $l_1$ norm of coherence. For arbitrary states, since the BLP measure provides the maximum value when optimized over all initial states, the coherence provides a lower bound to it, i.e., 
\begin{equation} \label{lowerbound}
    \mathcal{B} \ge \max_{\rho_{1}(0)} \int_{\frac{d}{dt} C_{l_1}(\rho_1(t))>0} \frac{d}{dt} C_{l_1}(\rho_1(t)) dt .
\end{equation} 
This suggests that, if at instant there is a revival of coherence from the environment to the system, $\frac{d}{dt} C_{l_1}(\rho_1(t)) > 0$, there is an information backflow from the environment to the system giving $\mathcal{B} > 0$. 
Moreover, in \cite{biswas2017interferometric,qureshi2019coherence} the authors have established that the visibility of an interference pattern maximised over interferometric unitaries is equal to the $l_1$ norm of coherence (this is also supported by the examples provided below). Hence, Eq.~\eqref{lowerbound} reduces to 
\begin{equation*} 
    \mathcal{B} \ge \max_{\rho_1(0)} \int_{\frac{d}{dt} \tilde{V}(t)>0} \frac{d}{dt} \tilde{V}(t) dt .
\end{equation*}
Thus, $\dot{\tilde{V}}(t)>0$ at some time $t$ implies $\mathcal{B} > 0$. This suggests that an increase in $\tilde{V}(t)$ at some time is a sufficient condition for non-Markovianity. 

Although the proof is presented for qubit systems, the BLP measure is always lower-bounded by the coherence extends to arbitrary higher dimensions. However, for $d \ge 3$, one must identify an appropriate pair of quantum states for which the trace distance is exactly equal to the $l_1$ norm of coherence.
\end{proof}
Below, we provide supporting examples to illustrate this. For the sake of convenience, we consider cases where $D_{cl}=0$, and hence Eq.~\eqref{equality} holds.
\begin{example}
    Pure dephasing dynamics: Consider the master equation given by Eq.~\eqref{dephasing} with the Lindblad coefficient given by 
    \begin{equation}
        \gamma(t)=\int_{0}^{\infty} J(\omega) \coth[{\frac{\omega}{2 k_B T}} ]\frac{\sin{\omega t}}{\omega} d\omega  \label{Lindbladcoefficient}
    \end{equation}
    where $J(\omega)$ is the spectral density function of the environment given by
    \begin{equation}
        J(\omega) = \frac{\omega^s}{\omega_c^{s-1}-1} e^{\frac{-\omega}{\omega_c}}. \label{spectraldensity}
    \end{equation}
    $s$ and $\omega_c$ represent the Ohmicity parameter and the cut-off frequency respectively. At a temperature $T = 0$, the Lindblad coefficient reduces to \cite{chanda2016delineating}
    \begin{equation}
        \gamma_{0} (t,s)= \frac{\omega_c^{s} \Gamma(s) \sin[{s \arctan(\omega_c t)}]}{[1+(\omega_c t)^2]^{s/2}}, 
    \end{equation}
    where $\Gamma(s)$ is the Euler Gamma function.
    This dynamics is non-Markovian, i.e., $\gamma_{0}(t,s) < 0$ for $s>2$ \cite{PhysRevA.87.010103}. If the initial state of the system is 
    \begin{equation}  \label{initial state}
        \rho(0)=\begin{pmatrix}
            \rho_{00}(0) & \rho_{01}(0) \\
            \rho_{10}(0)  & \rho_{11}(0) 
        \end{pmatrix},
    \end{equation} the evolution of this state under the dephasing dynamics is governed by the Kraus operators given by $K^{\text{dep}}_0=\sqrt{1-\frac{p(t)}{2}} I$ and $K^{\text{dep}}_1=\sqrt{\frac{p(t)}{2}} \sigma_z$. The time-evolved state is given by
    \begin{equation*}
        \rho_{\text{dep}}(t)=\begin{pmatrix}
            \rho_{00}(0) & \rho_{01}(0) G_{\text{dep}}(t) \\
            \rho_{10}(0) G_{\text{dep}}(t) & \rho_{11}(0) 
        \end{pmatrix},
    \end{equation*}
    where $G_{\text{dep}}(t)=1-p(t)=e^{-2\int_{0}^{t} \gamma(t^{'})dt^{'}}$ is the decoherence function for the dephasing dynamics. The $l_1$ norm of coherence is then given by  $$C_{l_1}(\rho_{\text{dep}}(t))=2|{\rho_{01}(0)}| G_{\text{dep}}(t).$$ Thus, $\frac{dC_{l_1}(\rho(t))}{dt}= 2|\rho_{01}(0)|\frac{d G_{\text{dep}}(t)}{dt}$. \\This shows that whenever $\gamma_0(t,s)<0$, or equivalently $s>2$, we have $\frac{dC_{l_1}(\rho(t))}{dt}>0$. This indicates the revival of coherence for a non-Markovian dynamics. Moreover, if we have a Mach-Zehnder interferometer with $\rho=\rho_{\text{dep}}(t)$ as the input to the interferometer, the visibility maximized over all interferometric unitaries is given by 
\begin{equation*}
\begin{split}
    \tilde{V}(t) = & \max_{U \in \mathcal{U}_{int}} \left|\operatorname{Tr}[U \rho_{\text{dep}} (t)]\right| \\ & = 2|{\rho_{01}(0)}| G_{\text{dep}}(t), 
    \end{split}
\end{equation*} which equals the $l_1$ norm of coherence. The corresponding optimal unitary is obtained for $\alpha=\phi, \beta=-\phi$ in Eq.~\eqref{Uint}, i.e., 
\begin{equation} \label{Uopt}
    U_{opt}= \begin{pmatrix}
            0 & e^{i \phi} \\
            e^{-i \phi}  & 0 
        \end{pmatrix}, 
\end{equation}
where $\phi= \arg[\rho_{01}(0)].$
For this value of $\tilde{V}(t)$, the BLP measure is given by 
\begin{equation}
    \mathcal{B}_{\text{dep}}= \int_{\gamma_{0} (t,s) < 0} \frac{d}{dt} G_{\text{dep}}(t) dt.
\end{equation}
The visibility-based measure of non-Markovianity ($\mathcal{B}$) is plotted as a function of the Ohmicity parameter ($s$) in Figure \ref{figdephasing}. From the Figure, it is evident that $\mathcal{B}_{\text{dep}}>0$ for $s>2$, indicating the signature of the dynamics being non-Markovian for $s>2$.
\begin{figure}[ht]
\includegraphics[width=.47\textwidth]{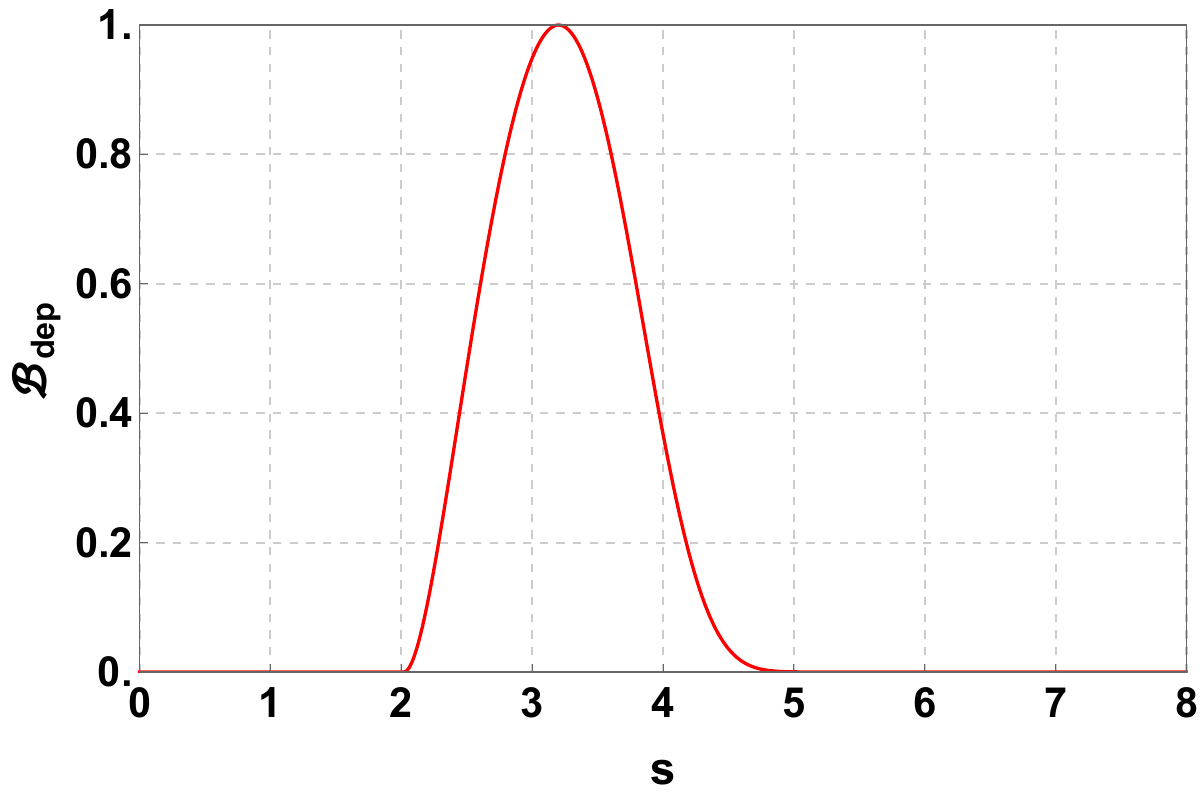} 
\caption{Visibility-based BLP measure of non-Markovianity as a function of the Ohmicity parameter for dephasing dynamics.}\label{figdephasing}
\centering
\end{figure}
\end{example}
\begin{example}
    Depolarizing dynamics: The master equation is given by Eq.~\eqref{depolarizing}. For simplicity, we assume $\gamma_{x}(t)=\gamma_{y}(t)=\gamma_{z}(t)=\gamma(t)$; nevertheless, the analysis remains valid even when these rates differ. The Lindblad coefficient, spectral density are given by Eqs.~\eqref{Lindbladcoefficient} and \eqref{spectraldensity} respectively. The corresponding Kraus operators are given by $$K^{\text{depol}}_0=\sqrt{1-p(t)}I, K^{\text{depol}}_i=\sqrt{\frac{p(t)}{3}}\sigma_i, i \in \{x,y,z\}.$$ Under this Kraus representation, the initial state (as given by Eq.~\eqref{initial state}) evolves to
    \begin{equation*}
    \begin{split}
        & \rho_{\text{depol}}(t)= \\ & \begin{pmatrix}
            (1-\frac{2p}{3})\rho_{00}(0)+\frac{2p}{3}\rho_{11}(0) & \rho_{01}(0) G_{\text{depol}}(t) \\
            \rho_{10}(0) G_{\text{depol}}(t) & \frac{2p}{3}\rho_{00}(0)+(1-\frac{2p}{3})\rho_{11}(0))
        \end{pmatrix},
        \end{split}
    \end{equation*}
        where $G_{\text{depol}}(t)=1-\frac{4 p(t)}{3}=e^{-4\int_{0}^{t} \gamma(t^{'})dt^{'}}$ is the decoherence function for the depolarizing dynamics. The measure of coherence for this evolved state is 
        $$C_{l_1}(\rho_{\text{depol}}(t))=2|{\rho_{01}(0)}| G_{\text{depol}}(t).$$ $$ \text{Hence}, \frac{dC_{l_1}(\rho(t))}{dt}= 2|\rho_{01}(0)|\frac{d G_{\text{depol}}(t)}{dt}.$$
        Similar to the previous example, $\gamma_{0}(t,s) < 0 (s > 2)$ indicates a revival of coherence. The visibility for an input state $\rho=\rho_{\text{depol}}(t)$ is then given by 
        \begin{equation*}
\begin{split}
    \tilde{V}(t) = & \max_{U \in U_{int}} \left|\operatorname{Tr}[U \rho_{\text{depol}} (t)]\right| \\ & = 2|{\rho_{01}(0)}| G_{\text{depol}}(t), 
    \end{split}
\end{equation*} which equals the $l_1$ norm of coherence. The corresponding optimal unitary is given by Eq.~\eqref{Uopt}.
For this value of $\tilde{V}(t)$, the BLP measure is given by 
\begin{equation}
    \mathcal{B}_{\text{depol}}= \int_{\gamma_{0} (t,s) < 0} \frac{d}{dt} G_{\text{depol}}(t) dt.
\end{equation}
For $s>2$, the depolarizing dynamics is non-Markovian, and hence $\mathcal{B}_{\text{depol}} > 0$. This is represented in Figure \ref{figdepolarizing}.
\begin{figure}[ht]
\includegraphics[width=.47\textwidth]{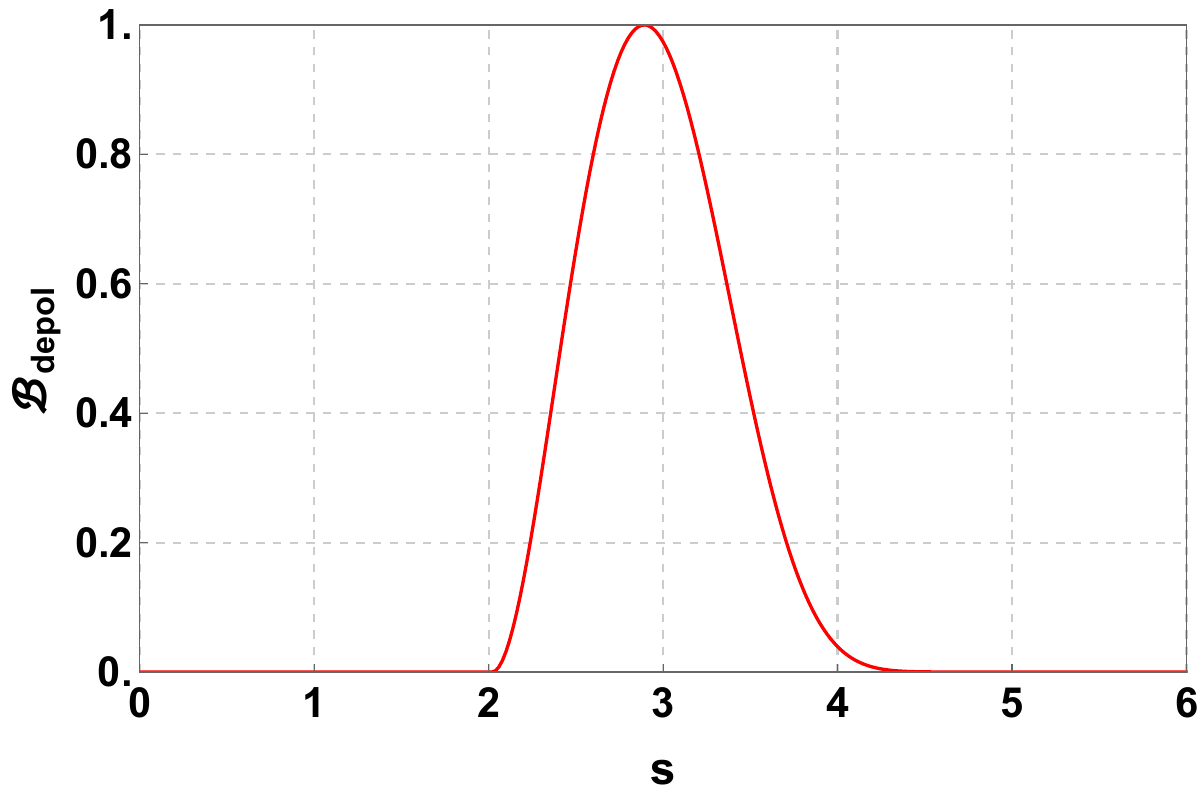}
\caption{Visibility-based BLP measure of non-Markovianity as a function of the Ohmicity parameter for depolarizing dynamics.}\label{figdepolarizing}
\centering
\end{figure}
\end{example}
In both the examples studied above, the visibility-based BLP measure attains its maximum value for $s \simeq 3.2$ and $s \simeq 2.8$, as evident from Figures \ref{figdephasing} and \ref{figdepolarizing} respectively. The decrease in the measure for higher values of $s$ is an indication of smaller amplitude of oscillations in the dephasing and depolarizing rates \cite{PhysRevA.90.052103}. \\
\section{Conclusion} \label{secIV}
The detection of memory effects in open quantum systems is significant due to the various advantages offered by such dynamics in different quantum information processing tasks \cite{laine2014nonlocal,PhysRevLett.100.180402,bylicka2014non,zhang2024non,PhysRevE.97.062108,marcantoni2017entropy}.
The different approaches to characterizing such dynamics reveal their diverse underlying features. In this work, we have developed an optical interferometric framework for witnessing non-Markovian dynamics. In particular, we have demonstrated that the proposed witness-based and coherence-based schemes can characterize such dynamics through the phase shift and visibility of the resulting interference pattern, respectively. The effectiveness of the proposed detection criteria was illustrated through representative examples of pure dephasing and depolarizing dynamics.

Our work introduces an operationally meaningful and experimentally implementable route for identifying memory effects. Unlike conventional mathematical approaches based on the breakdown of CP divisibility of the Choi operator, our framework provides a more physically intuitive characterization of memory effects through directly measurable parameters of the interference pattern. The generality and experimental accessibility of the framework make it potentially useful for applications involving quantum-enhanced technologies, particularly in scenarios where environmental noise and finite-size quantum devices play an important role. 

While the present work focuses on qubit dynamics, extending the detection of memory effects to multiqubit systems through suitably designed multi-particle interferometric setups presents a promising avenue for future research. Furthermore, developing non-Markovian measures based on visibility and establishing their connections with other prominent measures of non-Markovianity, including those based on mutual information, accessible information, uncertainty relations, remains an important open problem.
\section{Acknowledgements} 
The authors acknowledge Himadri Sekhar Dhar for insightful discussions and comments on the manuscript. A.G.M. acknowledges the financial support provided by the IIT Goa Startup Grant (No. 2025/SG/AGM/058).
\bibliography{main}

\end{document}